\documentclass{cccg20}
\usepackage{graphicx,amssymb,amsmath}
\usepackage{hyperref,aliascnt}
\usepackage{balance}

\newaliascnt{lemma}{theorem} \newtheorem{lemma}[lemma]{Lemma}
\aliascntresetthe{lemma} 

\newaliascnt{corollary}{theorem}

\aliascntresetthe{corollary}

\title{Dynamic Products of Ranks}

\author{David Eppstein\thanks{Computer Science Department,
        University of California, Irvine, {\tt eppstein@uci.edu}. This work was supported in part by the US National Science Foundation under grant CCF-1616248.}}

\index{Eppstein, David}

\begin{document}
\thispagestyle{empty}
\maketitle

\begin{abstract}
We describe a data structure that can maintain a dynamic set of points given by their Cartesian coordinates, and maintain the point whose product of ranks within the two coordinate orderings is minimum or maximum, in time $O(\sqrt{n\log n})$ per update.
\end{abstract}

\section{Introduction}

The \emph{rank} of an element in a collection of elements is its position in a list of all elements, sorted by some associated numerical value. If elements have a multidimensional vector of values associated with them, then each of these values gives rise to a different rank, and we may wish to aggregate these multiple ranks into a single combined score. One common method of aggregating ranks is to use the geometric mean or equivalently the product of ranks as the combined score. This method is used in applications ranging from finding differentially regulated genes in DNA microarray data~\cite{BreArmAmt-FEBS-04}, choosing winners in multi-discipline sports events~\cite{IFSC}, and measuring the scholarly output of economists~\cite{Zim-Eco-13} to image recognition~\cite{ChaBar-ICASSP-06} and spam filtering in web search engines~\cite{DwoKumNao-WWW-01}.

In many of these applications, it is natural for the elements in the collection and their associated numerical values to change dynamically, and when they do the whole system of ranks for other elements may change. For instance, inserting one new element, with a low numerical value, will increase the ranks of all elements with larger values. This raises the question: how can we update the elements and their numerical values, and maintain information about the product of ranks?

We can model this as a geometry problem, in which the elements in the collection are modeled as points in the Cartesian plane, with the $x$- and $y$-coordinates of these points representing their associated numerical values.
In this model, we would like to maintain a dynamic set of pairs of real numbers, subject to point insertion and point deletion, and as we do so, maintain dynamically the point whose product of ranks in the two coordinate orderings is minimum or maximum.

In this work we provide a solution to this dynamic product of ranks problem. We solve the dynamic product of ranks problem, in the special case when there are two rankings being combined,
 in time $O(\sqrt{n\log n})$ per update.

There are three main ideas to our method:
\begin{itemize}
\item We partition the points into \emph{rigid} subsets: sets of points whose ranks all change in lockstep with each operation (that is, without changing the difference between the ranks of any two elements in the set). Our partition will have the property that each update will rebuild rigid subsets of total size $O(\sqrt{n\log n})$ and search for the point with minimum or maximum product of ranks within $O(\sqrt{n/\log n})$ of these subsets.
\item We provide two solutions to the dynamic product of ranks problem within each rigid subset. One solution applies a lifting transformation (to the pairs of ranks of the points, not their given coordinates) to turn it into a problem of querying a (static) three-dimensional convex hull. Dually, the other solution uses analogues of the classical Voronoi diagram and farthest-point Voronoi diagram, minimization and maximization diagrams with convex-polygon cells.
\item We provide linear-time constructions for the lifted convex hull in the minimization version of the problem, and for the maximization diagram in the maximization version of the problem, adapted from two different algorithms for linear-time construction of Voronoi diagrams of points in convex position.
\end{itemize}

Our method can be generalized to larger numbers of rankings, but with a quadratic blowup in the dimension of the lifting transformation that (together with the high complexity of higher-dimensional extreme point queries) leads to a running time per update that is only slightly smaller than the trivial naive algorithm of updating all rankings and recomputing all products in linear time per update. For this reason, we restrict our attention to maintaining information about the product of two rankings.

\section{Rigid subsets}

\subsection{Lifted hull}
We say that a subset $S$ of elements in our product of ranks problem is \emph{rigid}, through a sequence of updates, if none of the updates performs an insertion or deletion of an element of $S$, or of another element whose position in either of the two rankings lies between two elements of $S$. Equivalently, the difference in ranks of any two elements of $S$ remains invariant throughout the given sequence of updates.

\begin{lemma}
\label{lem:rigid}
Let $S$ be any subset of elements in the product of ranks problem, of size $m$. Then in time $O(m\log m)$ we can build a data structure for $S$ such that, throughout any sequence of updates for which $S$ is rigid, we can compute the elements of $S$ with the minimum or maximum product of ranks in time $O(\log m)$ per update.
\end{lemma}

\begin{proof}
Let $(x_i,y_i)$ be the ranks of the elements of $S$ prior to the sequence of updates for which $S$ is rigid.
We construct the three-dimensional convex hull of the lifted points $(x_i,y_i,x_iy_i)$, and a Dobkin--Kirkpatrick hierarchy allowing us to perform linear optimization queries (finding the extreme point on the resulting hull of a given linear function) in time $O(\log m)$ per query~\cite{DobKir-Algs-85}. The hull takes $O(n\log n)$ time to construct and its Dobkin--Kirkpatrick hierarchy takes an additional $O(n)$ time. For each element, let $z_i=x_iy_i$ denote its third coordinate in this lifted point set.

After a sequence of updates that have changed the ranks by subtracting the same offset $a$ from each rank $x_i$ and the same offset $b$ from each rank $y_i$ within $S$, the updated products of ranks are
\[(x_i-a)(y_i-b)=ab-ay_i-bx_i+z_i,\]
a linear function of the three coordinates of the lifted points, so the elements with the minimum and maximum product of ranks can be found by a linear optimization query.
\end{proof}

This method is closely analogous to the classical lifting transformation of two-dimensional closest-point problems to three-dimensional extreme-point problems~\cite{Bro-IPL-79}, which in its most commonly used form maps pairs $(x_i,y_i)$ to triples $(x_i,y_i,x_i^2+y_i^2)$; however, we use a different quadratic function for the third coordinate.
Note that we will only query this structure for pairs $(a,b)$ with $a\le x_i$ and $b\le y_i$, because the differences $x_i-a$ and $y_i-b$ represent ranks and are therefore non-negative.

\subsection{Linear time construction}

To construct the lifted hull more quickly, it is helpful to reduce the set of points to a subset whose projection to the plane is convex.

\begin{figure}[t]
\centering
\includegraphics[scale=0.45]{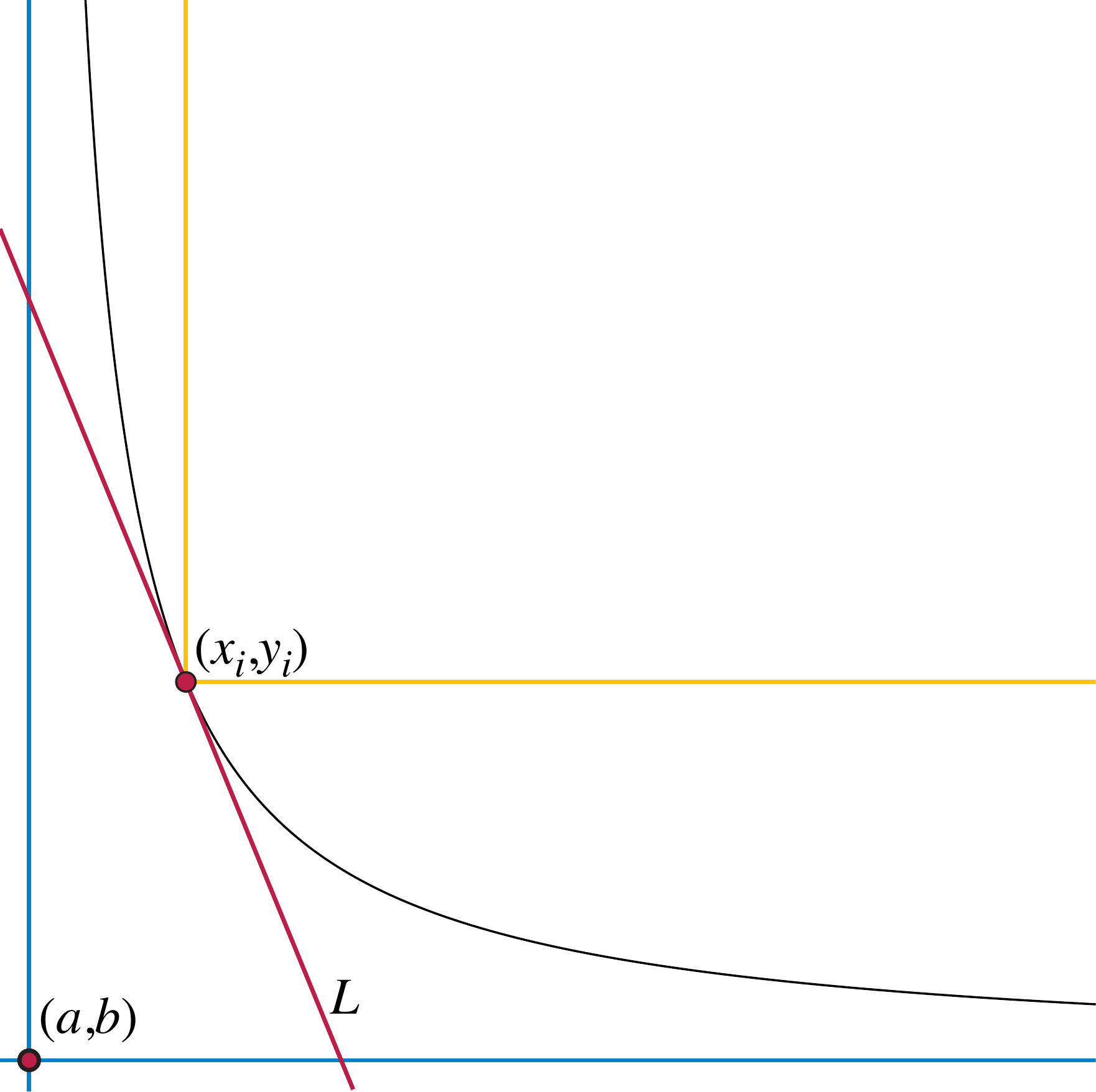}
\caption{The minimizer of $(x_i-a)(y_i-b)$ must be a convex hull vertex, because the region below line $L$, the tangent to the hyperbola through $(x_i,y_i)$, must be disjoint from $S$ (\autoref{lem:hull}). Analogously, the maximizer of $(x_i-a)(y_i-b)$ must be a maximal point of $S$, because the region above and to its left (yellow) must be disjoint from $S$ (\autoref{lem:dom}).}
\label{fig:hyperbola}
\end{figure}

\begin{lemma}
\label{lem:hull}
Let $S$ be a set of points, let $(a,b)$ be a pair of numbers with $a$ less than or equal to all $x$-coordinates in $S$ and $b$ less than or equal to all $y$-coordinates in $S$. Let $(x_i,y_i)$ be the point in $S$ minimizing $(x_i-a)(y_i-b)$.
Then $(x_i,y_i)$ lies on the convex hull of $S$.
\end{lemma}

\begin{proof}
The locus of points $(x,y)$ with $(x-a)(y-b)=(x_i-a)(y_i-b)$ is a hyperbola, asymptotic to the lines $x=a$ and $y=b$, with $(x_i,y_i)$ on its positive branch. Let $L$ be the line tangent to this hyperbola at $(x_i,y_i)$; see \autoref{fig:hyperbola}. Then the halfplane below $L$ must be disjoint from $S$, for any point $(x_j,y_j)$ between $L$ and the other branch of the hyperbola would have a smaller value of $(x_j-a)(y_j-b)$ and by the assumptions on $a$ and $b$ there are no points of $S$ on the other side of the other branch of the hyperbola.
\end{proof}

Aggarwal et al.~\cite{AggGuiSax-DCG-89} showed that, for 3d points whose two-dimensional projection is convex, the 3d convex hull can be constructed in linear time. In the next lemma we apply this result to the lifted hull of \autoref{lem:rigid}.

\begin{lemma}
\label{lem:linmin}
Let $S$ be any subset of elements in the product of ranks problem, of size $m$, for which the sorted order by $x$-coordinate is known. Then in time $O(m)$ we can build a data structure for $S$ such that, throughout any sequence of updates for which $S$ is rigid, we can compute the elements of $S$ with the minimum product of ranks in time $O(\log m)$ per update.
\end{lemma}

\begin{proof}
We use Graham scan to compute the 2d convex hull from the sorted order of points in linear time, and the algorithm of Aggarwal et al.~\cite{AggGuiSax-DCG-89} to compute the 3d convex hull from the 2d convex hull in linear time. The Dobkin--Kirkpatrick hierarchy construction time is also linear.
\end{proof}

\subsection{Maximization diagram}

Instead of lifting the points $(x_i,y_i)$ to the convex hull of three-dimensional points $(x_i,y_i,x_iy_i)$, an alternative representation for each rigid subset would be to represent it by the minimization diagram or maximization diagram of the functions $f_i(a,b)=(x_i-a)(y_i-b)=ab-ay_i-bx_i+z_i$. Then, the minimum or maximum product of ranks for the rigid subset with rank offsets $a$ and $b$ could be obtained by performing a point location query in this diagram, rather than by performing an extreme-point query on a three-dimensional polyhedron.

Because the quadratic term $ab$ in the definition of the function $f_i(a,b)$ does not depend on the point $(x_i,y_i)$, and is equal for all points, it does not affect the minimization or maximization: we obtain the same minimization or maximization diagrams for the linear functions $g_i(a,b)=-ay_i-bx_i+z_i$. As the minimization or maximization diagram of linear functions, these diagrams have convex polygon cells, separated by \emph{bisector} lines, the lines consisting of the points $(a,b)$ at which two of these functions are equal.

\begin{figure}[t]
\centering
\includegraphics[width=0.7\columnwidth]{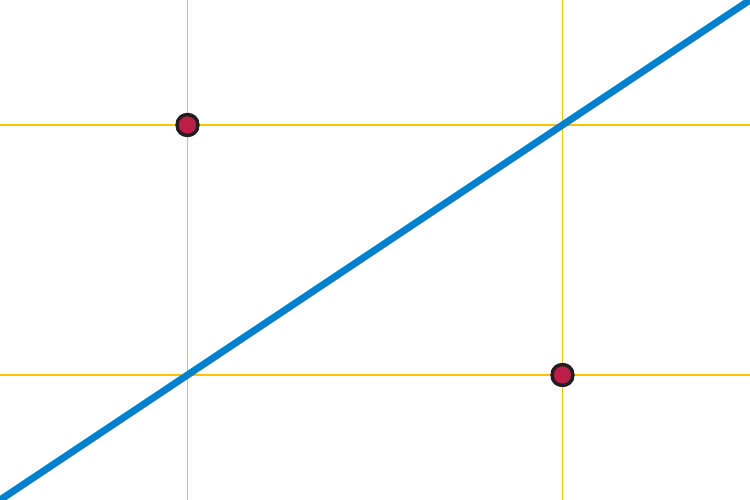}
\caption{The bisector between two sites in the minimization diagram is the line through the other two corners of their bounding box.}
\label{fig:bisector}
\end{figure}

\begin{lemma}
\label{lem:bisector}
The bisector of any two given points (sites) $(x_i,y_i)$ and $(x_j,y_j)$ in the minimization or maximization diagram described above is a line that passes through the other two corners $(x_i,y_j)$ and $(x_j,y_i)$ of the bounding box of the two points (\autoref{fig:bisector}).
\end{lemma}

\begin{proof}
When the bounding box is a square, this follows by symmetry: a reflection through the line described in the lemma maps the two given points to each other, swapping the two Cartesian coordinates, so for any point $(a,b)$ on the line described in the lemma, the coordinate differences between $(a,b)$ and the two given points are equal but reversed.
That is, $|x_i-a|=|y_j-b|$ and $|x_j-a|=|y_i-b|$. Since the quantity being minimized is the product of these coordinate differences, it is equal for the two given points along this line.

For any other two points, not both on the same vertical or horizontal line, we may apply a linear transformation to one of the coordinates that makes the bounding box a square;
this transformation affects both of the functions $g_i$ and $g_j$ in the same way, so the bisector of the transformed points (the diagonal of the square) is the transformation of the bisector, which must therefore be the diagonal of the original bounding box. The remaining case, that the points are on a horizontal or vertical line, follows by continuity.
\end{proof}

\begin{figure}[t]
\includegraphics[width=\columnwidth]{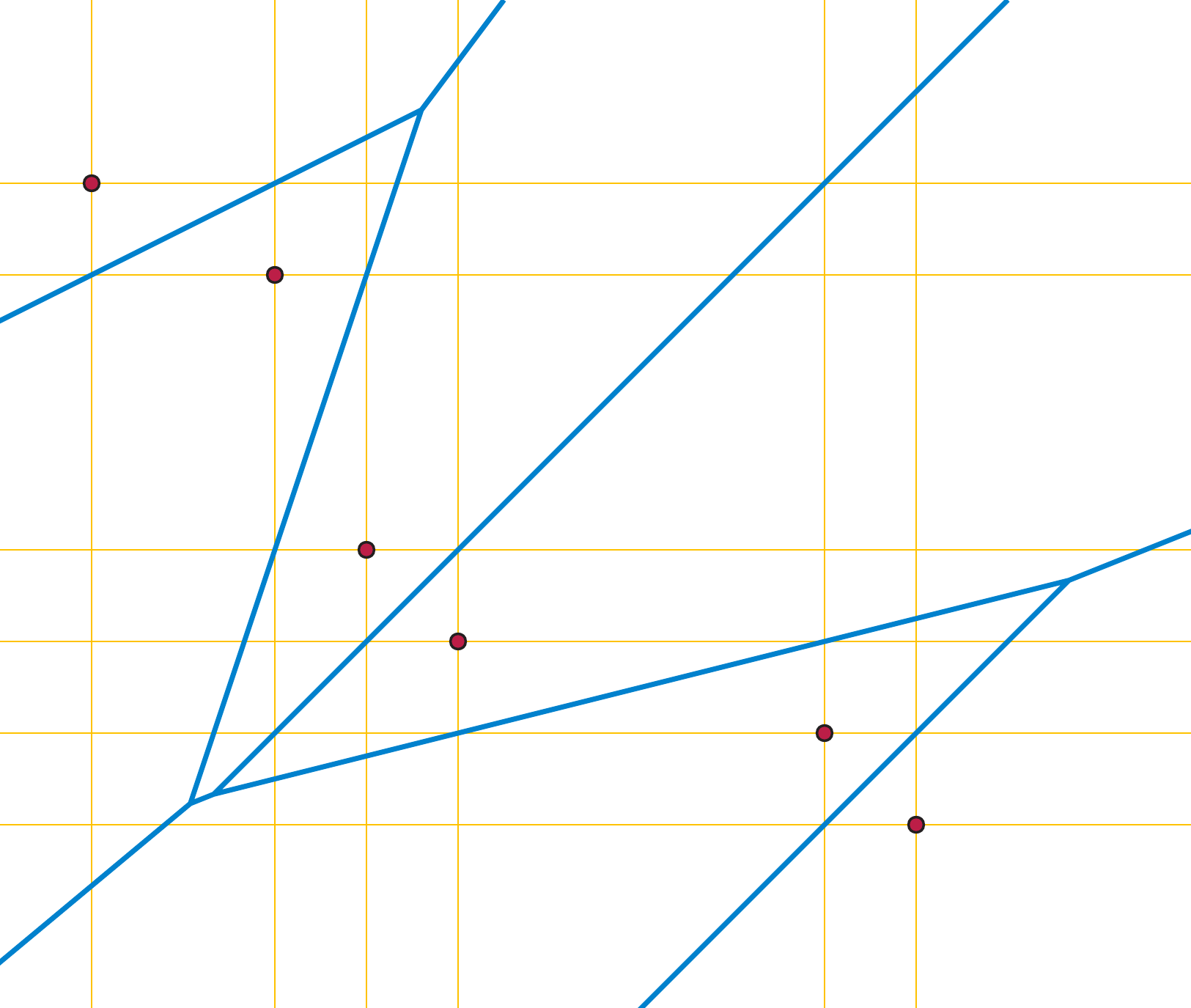}
\caption{The maximization diagram for a given set of maximal points. Although this diagram is well-defined over the whole plane, we will only query it within the bottom-left quadrant, below and to the left of all the given points.}
\label{fig:diagram}
\end{figure}

\autoref{fig:diagram} depicts an example of the maximization diagram described above.

\subsection{Expected linear time construction}

These diagrams can be constructed in $O(n\log n)$ time, either by interpreting them as a lower or upper envelope of three-dimensional planes (the graphs of the functions they minimize or maximize) or by using algorithms for abstract Voronoi diagrams with bisectors determined as in \autoref{lem:bisector}~\cite{MehMeiODu-DCG-91}. 
However, as we now show, they can be constructed in expected linear time.

Our construction begins with the following analogue of \autoref{lem:hull}.
We observe that, in constructing the maximization diagram for a collection of points, we need only include the points $(x_i,y_i)$ that are maximal (meaning that there is no other point $(x_j,y_j)$ with $x_j\ge x_i$ and $y_j\ge y_i$), for those are the only ones that can produce the maximum of the function values at any point $(a,b)$.

\begin{lemma}
\label{lem:dom}
Let $S$ be a set of points, let $(a,b)$ be a pair of numbers with $a$ less than or equal to all $x$-coordinates in $S$ and $b$ less than or equal to all $y$-coordinates in $S$. Let $(x_i,y_i)$ be the point in $S$ maximizing $(x_i-a)(y_i-b)$.
Then $(x_i,y_i)$ is one of the maximal points of $S$, meaning that there is no other point $(x_j,y_j)$ in $S$ with $x_j\ge x_i$ and $y_j\ge y_i$.
\end{lemma}

\begin{proof}
Any such point $(x_j,y_j)$ would have a larger value of $(x_i-a)(y_i-b)$.
\end{proof}

The quarter-plane of points with larger $x$- and $y$-coordinates than $(x_i,y_i)$, and their relation to the hyperbola of points with equal query values to $(x_i,y_i)$, is shown in \autoref{fig:hyperbola}.

To construct the maximization diagram in expected linear time we adapt an algorithm by Paul Chew for Voronoi diagrams of convex polygons~\cite{Che-90}.

\begin{lemma}
\label{lem:maxadj}
Let $S$ be a set of points, all of which are maximal in $S$, indexed in sorted order by their $x$-coordinates, and let $(x_i,y_i)$ and $(x_{i+1},y_{i+1})$ be consecutive points in this ordering. Then in the maximization diagram for $(x_i-a)(y_i-b)$, the cells for these two points share an edge.
\end{lemma}

\begin{proof}
Within the bounding rectangle of $(x_i,y_i)$ and $(x_{i+1},y_{i+1})$, the point $(x_i,y_i)$ has a larger query value than all points of $S$ with smaller index, and the point $(x_{i+1},y_{i+1})$ has a larger query value than all points of $S$ with larger index, so the maximization diagram within the rectangle consists only of points in the cells for these two points. By \autoref{lem:bisector} the cells meet within the rectangle along the bisector of these two points, which is the diagonal of the rectangle.
\end{proof}

The shared edge is not in a part of the diagram that we will query in our data structure for products of ranks, but its location is unimportant for the use we will make of it in the following lemma.

\begin{lemma}
\label{lem:linmax}
Let $S$ be any subset of elements in the product of ranks problem, of size $m$, for which the sorted order by $x$-coordinate is known. Then in randomized expected time $O(m)$ we can build a data structure for $S$ such that, throughout any sequence of updates for which $S$ is rigid, we can compute the elements of $S$ with the maximum product of ranks in time $O(\log m)$ per update.
\end{lemma}

\begin{proof}
The maximal points in $S$ can be found in linear time from the sorted order by $x$-coordinates, using a stack algorithm closely related to Graham scan.

We construct the maximization diagram by a randomized incremental algorithm in which we randomly permute the points and add them to the diagram one at a time in that random order. By the analysis of Chew~\cite{Che-90}, this can be done in expected constant time per point as long as we know the identity of a neighboring cell in the diagram of the points added so far. We can form a random permutation with this additional information about neighboring cells by starting with a doubly linked list of all of the points, in $x$-coordinate order,  deleting randomly chosen points from the linked list until none are left, and then reversing the order of the deletions. By \autoref{lem:maxadj}, the neighbors of a point $(x_i,y_i)$ in the linked list at the time of its deletion will form neighboring cells in the maximization diagram at the time of its insertion.

Because it is the maximization diagram of a set of linear functions, we can interpret this diagram as a three-dimensional intersection of halfspaces, and construct a Dobkin--Kirkpatrick hierarchy from it in linear time, suitable for performing point location queries in logarithmic time. (Alternatively, the history DAG of a vertical decomposition of the randomized incremental maximization diagram construction can be used as a point location data structure with logarithmic expected time per query.)
\end{proof}

\section{Partitioned data structure}
\subsection{One-dimensional partition}

To partition our given elements into rigid subsets, we first consider a one-dimensional partition method, which we will apply separately to the two rankings of the elements.

\begin{lemma}
\label{lem:1d}
Let $f$ be any positive concave function of a single argument. Then for any sequence $S$ of ordered values undergoing insertions and deletions, we can maintain a partition of $S$ into an ordered sequence of $O(n/f(n))$ contiguous subsets, with $O(f(n))$ elements in each subset,
changing $O(1)$ subsets per update, using a data structure with time $O(\log n)$ per update, where $n$ denotes the current size of $S$.
\end{lemma}

\begin{proof}
We use binary search trees to keep track of the sequence of elements and the sequence of subsets. As keys for the binary search tree of subsets, we use the values of their first elements. In this way we can find the subset containing the updated element, after any update, and determine the new size of this subset. We also keep track of the sizes of each subset and maintain priority queues for the largest subsets and for the smallest consecutive pairs of subsets.

We maintain as an invariant the requirements that the sizes of all subsets in the partition are at most $2\lceil f(n)\rceil+2$, and that no two consecutive subsets both have size less than $\lceil f(n)\rceil-1$. We say that our structure is \emph{growing} if, for the most recent update having a different value of $\lceil f(n)\rceil$, that value was smaller than 
the current value, and \emph{shrinking} otherwise. If the structure is growing, we require that all subsets have size at most $2\lceil f(n)\rceil$, and if it is shrinking, we require that no two consecutive subsets both have size less than $\lceil f(n)\rceil$.

On each update, if the structure is growing, we select an arbitrary pair of consecutive subsets of size $\lceil f(n)\rceil-1$ (if such a pair exists) and merge them into a single subset. If the structure is shrinking, we select an arbitrary subset of size greater than $2\lceil f(n)\rceil$ (if such a subset exists) and split it into two subsets of size as close to equal as possible. We claim that this is sufficient to maintain our invariants. Clearly, it does so at the updates for which $\lceil f(n)\rceil$ does not change, so we need only consider the steps at which it does change.

In the case that $\lceil f(n)\rceil$ changes in such a way that the structure was growing before the update and is shrinking after the update, the invariants are automatically maintained, because the ranges of sizes of subsets and consecutive pairs of subsets that are allowed remain unchanged. The same is true when the structure was shrinking before the update and growing after the update.

When $\lceil f(n)\rceil$ increases twice in a row (so that it was growing both before and after the second increase), let $n_0$ be the value of $n$ at the first increase. Then at that time, there must be at most $n_0/f(n_0)$ consecutive pairs of small subsets, and (by concavity of $f$) at least $n_0/f(n_0)$ steps between the two increases. It only takes $n_0/2f(n_0)$ steps to eliminate all of the consecutive pairs of small subsets. So by the time that the second increase happens, all of the consecutive pairs of small subsets will have been eliminated, maintaining the invariant. 

Similarly, when $\lceil f(n)\rceil$ decreases twice in a row (so that it was shrinking both before and after the second decrease), let $n_0$ be the value of $n$ at the first decrease. Then at that time, there must be at most $n_0/2f(n_0)$ large subsets, and (by concavity of $f$) at least $n_0/f(n_0)$ steps between the two decreases. It only takes $n_0/2f(n_0)$ steps to eliminate all of the large subsets. So by the time that the second decrease happens, all of the consecutive pairs of large subsets will have been eliminated, maintaining the invariant.
\end{proof}

\subsection{Two-dimensional partition}

We now use our one-dimensional rank partition to partition the given elements into subsets, most of which remain rigid in each update. If the ranks of each element are $(x_i,y_i)$, we will maintain one rank partition on the ranks $x_i$, and a second rank partition on the ranks $y_i$, each with parameter $f(n)=\sqrt{n\log n}$. Then each subset $S_k$ of our two-dimensional partition will consist of elements that are grouped together both in the partition on the $x$-ranks and in the partition on the $y$-ranks.

\begin{lemma}
\label{lem:2d}
The partition into subsets $S_k$ described above has the following properties:
\begin{itemize}
\item There are $O(n/\log n)$ subsets.
\item Each update to the data causes $O(\sqrt{n/\log n})$ of the subsets, with total size $O(\sqrt{n\log n})$, to be non-rigid.
\item Each update to the data causes $O(\sqrt{n/\log n})$ of the subsets, with total size $O(\sqrt{n\log n})$, to be replaced by new subsets due to the change in the underlying one-dimensional partitions.
\end{itemize}
\end{lemma}

\begin{proof}
It follows from \autoref{lem:1d} and our choice of the function $f$ that each one-dimensional partition has $O(\sqrt{n/\log n})$ subsets, of size $O(\sqrt{n\log n})$, and that each update causes $O(1)$ changes to the one-dimensional partition.
Because each subset in the two-dimensional partition is determined by a pair of subsets in the two one-dimensional partitions, there are  $O(n/\log n)$ subsets in the two-dimensional partition.

In any update, only one subset of each one-dimensional partition contains non-rigid subsets of the two-dimensional partition. Therefore, the total number of non-rigid subsets is at most twice the number of two-dimensional subsets that can be contained in a single one-dimensional subset, $O(\sqrt{n/\log n})$, and the total size of the non-rigid subsets is at most twice the size of a one-dimensional subset, $O(\sqrt{n\log n})$. The analysis of the number of subsets that are replaced with new subsets and their total size is similar: each change to a one-dimensional subset causes changes to $O(\sqrt{n/\log n})$ two-dimensional subsets having a total of $O(\sqrt{n\log n})$ elements, so the bounds on replaced subsets follow from the fact that each update causes $O(1)$ changes to the one-dimensional partitions.
\end{proof}

\section{Which subsets to query?}

We introduced
\autoref{lem:hull} and \autoref{lem:dom} to aid in the efficient construction of rigid subsets, but they can
also be used to reduce the number of rigid subsets that we must query after any update.
As these two lemmas show, the point with the smallest product of ranks must be minimal in the coordinate ordering of the points, and the point with the largest product of ranks must be maximal.
The two-dimensional partition of \autoref{lem:2d} partitions the points in a grid pattern, and we need only query the rigid subsets for cells in this grid that can contain minimal or maximal points.

\begin{figure}[t]
\includegraphics[width=\columnwidth]{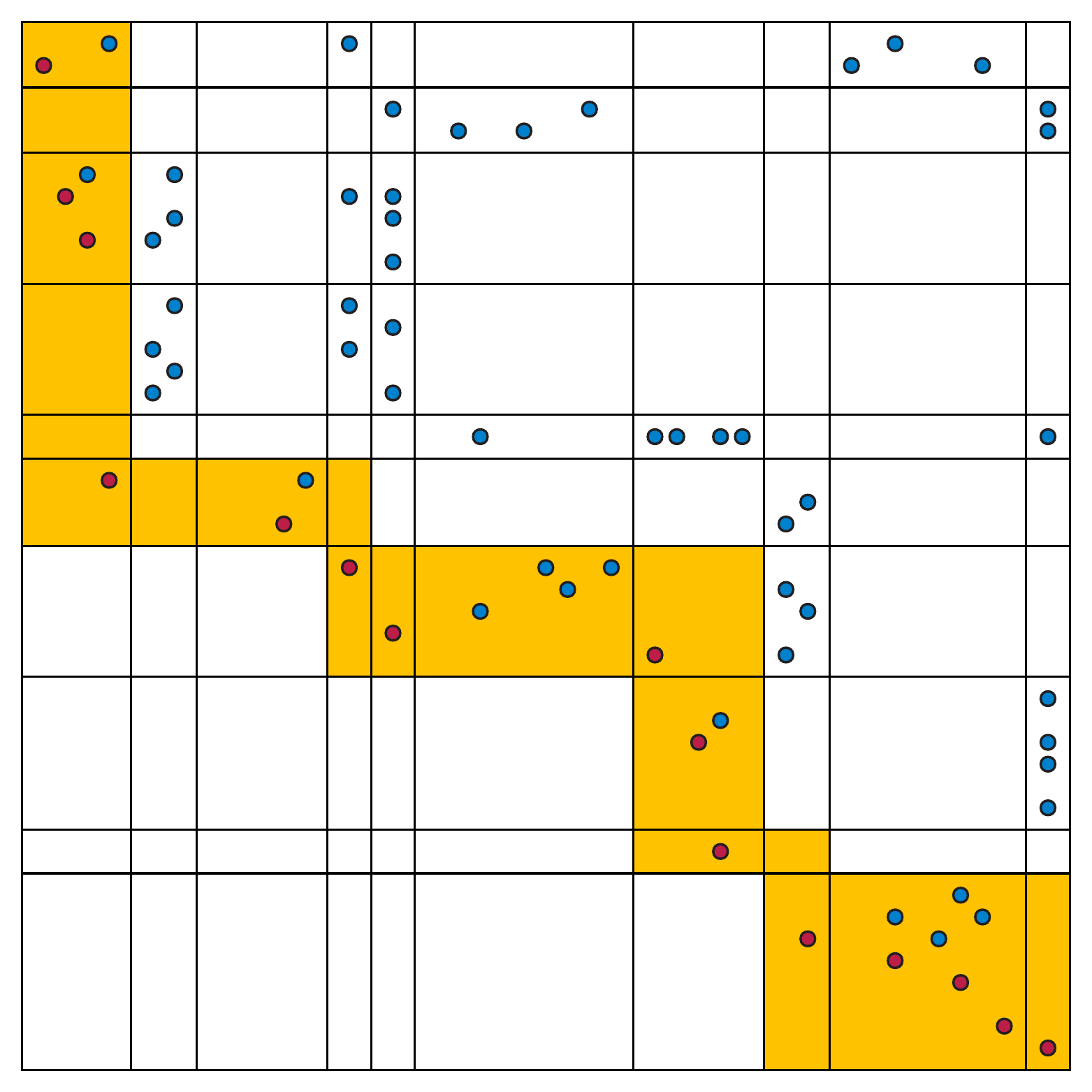}
\caption{A grid partition of a point set, and a path (yellow shading) through the cells of the grid, such that the cells of the path contain all minimal points of the set (shown as red).}
\label{fig:gridmin}
\end{figure}

\begin{lemma}
\label{lem:gridmin}
Let a given set of points be partitioned by $k$ axis-parallel lines into a grid of cells, represented in such a way that in constant time we can find the neighboring cell in any direction from any given cell and find the lowest nonempty cell in any column of the grid. Then in time $O(k)$ we can identify a subset of $O(k)$ of the grid cells that contain all of the minimal points in the set.
\end{lemma}

\begin{proof}
As we describe below, we select cells in the grid along a path from top left to bottom right, such that every unselected cell below the path is also below the lowest nonempty cell in its column, and every unselected cell above the path has a nonempty selected cell below and to the left of it. In this way, every minimal point of the given point set belongs to a selected cell, for there can be no points below and to the left of the path, and all points above and to the right are not minimal. \autoref{fig:gridmin} shows an example.

To find this path of grid cells, we begin at the top left cell of the grid. Then we repeatedly step to a neighboring cell, according to the following rules:
\begin{itemize}
\item If the current cell is the bottom right cell of the grid, we terminate the path.
\item If the current cell is not the lowest nonempty cell in its column, or if it belongs to the rightmost column, we step to the next cell down.
\item Otherwise, we step to the next cell to the right.
\end{itemize}
The path must extend across all columns, for it can only stop in the rightmost column.
If a cell is below the path, it must also be below the lowest nonempty cell in its column, or we would have stepped downward to it when the path crossed its column; therefore, all cells below the path are empty.
If a cell is above the path, then the path must have stepped below it in some column to the left of it, which can only happen when the lowest nonempty cell in that column is below and to the left of the given cell. Therefore, all cells above the path have a nonempty cell below and to the left of them.
\end{proof}

A similar method, with the ability to find the highest nonempty cell in each column, can find a path of grid cells containing all maximal points.

\section{Overall data structure}

Our overall data structure consists of:
\begin{itemize}
\item Two binary search trees on the two coordinate values of the elements, augmented to allow the rank of any element at any step of the update sequence to be looked up in logarithmic time per query.
\item Two one-dimensional partitions of the elements, one for each of the two rankings of the elements, according to \autoref{lem:1d}, with the parameter choice specified for \autoref{lem:2d}. 
\item The two-dimensional partition of the elements into rigid subsets $S_k$ defined from these one-dimensional partitions, according to \autoref{lem:2d}.
\item A graph describing the relation between neighboring cells in this two-dimensional partition, and the lowest or highest nonempty cell in each column of cells, suitable for use in \autoref{lem:gridmin}.
\item A sorted list of points in each partition set, sorted by their $x$-coordinates.
\item A data structure for maintaining the extreme points for the product of ranks of each subset $S_k$, through updates for which it is rigid, according to \autoref{lem:rigid}.
\end{itemize}

\begin{theorem}
The data structure described above can maintain the minimum or maximum product of ranks in time $O(\sqrt{n\log n})$ per update for the minimimum, or the same time bound in expectation for the maximum.
\end{theorem}

\begin{proof}
By \autoref{lem:2d}, each update causes changes to subsets $S_k$ of total size $O(\sqrt{n\log n})$; by \autoref{lem:linmin} and \autoref{lem:linmax}, reconstructing the extreme-point data structures for these subsets takes the stated time per update. After each update, we may use \autoref{lem:gridmin} to find a subset of $O(\sqrt{n/\log n})$ subsets to query, use the binary search trees to determine the offsets in rank for each of these selected subsets, and then query the extreme point within each subset in time $O(\log n)$ by \autoref{lem:rigid}. The total time for these queries is again the stated time per update. Maintaining the binary search trees and one-dimensional partitions takes an amount of time that is negligible with respect to this total time bound.
\end{proof}

\balance
\raggedright
\bibliographystyle{plainurl}
\bibliography{prodrank}

\end{document}